\documentclass[a4paper,twoside,twocolumn,11pt]{article}

\usepackage{cmap}
\usepackage[utf8]{inputenc}
\usepackage[T2A,T1]{fontenc}
\usepackage[russian,english]{babel}

\usepackage{amsfonts,amsmath,amssymb,amsthm}
\usepackage{graphicx}

\usepackage[hidelinks]{hyperref}
\usepackage{caption}
\usepackage{fancyhdr}
\usepackage[margin=18mm,top=28mm]{geometry}

\usepackage[ruled,lined,vlined,linesnumbered]{algorithm2e}

\renewenvironment{abstract}
 {\small
  \begin{center}
  \bfseries \abstractname\vspace{-.5em}\vspace{0pt}
  \end{center}
  \list{}{%
    \setlength{\leftmargin}{10mm}
    \setlength{\rightmargin}{\leftmargin}%
  }%
  \item\relax}
 {\endlist}

\usepackage{fancyvrb}
\newcounter{rlstno}

\newcommand{\en}{\selectlanguage{english}}

\en
\newtheorem{theorem}{\iflanguage{russian}{Теорема}{Theorem}}
\newtheorem{definition}{\iflanguage{russian}{Определение}{Definition}}
\newtheorem{remark}{\iflanguage{russian}{Замечание}{Note}}

\newtheorem{lemma}{\iflanguage{russian}{Лемма}{Lemma}}
\author{
D.V.~Luciv, D.V.~Koznov, A.A.~Shelikhovskii, K.Yu.~Romanovsky,\\
G.A.~Chernishev, A.N.~Terekhov, D.A.~Grigoriev, A.N.~Smirnova,\\ D.~V.~Borovkov, A.~I.~Vasenina\\
~Saint Petersburg State University\\
E-mail:~\texttt{\{d.lutsiv, d.koznov\}@spbu.ru; tshel231@gmail.com;}\\
\texttt{\{k.romanovsky, g.chernyshev, a.terekhov, d.a.grigoriev\}@spbu.ru;}\\ \texttt{anna.en.smirnova@gmail.com; danila\_yumsh@mail.ru;}\\
\texttt{saibrog@yandex.ru}
}
\title{Interactive Duplicate Search in Software Documentation}

\newcommand\blfootnote[1]{%
\begingroup
\renewcommand\thefootnote{}\footnote{\hspace{-5.9mm}#1}%
\addtocounter{footnote}{-1}%
\endgroup
}

\date{}
\begin{document}

\twocolumn[{%
\maketitle
\thispagestyle{empty}
\begin{abstract}
Various software features such as classes, methods, requirements, and tests often have similar functionality. This can lead to emergence of duplicates in their descriptive documentation. Uncontrolled duplicates created via copy/paste hinder the process of documentation maintenance. Therefore, the task of duplicate detection in software documentation is of importance. Solving it makes planned reuse possible, as well as creating and using templates for unification and automatic generation of documentation. In this paper, we present an interactive process for duplicate detection that involves the user in order to conduct meaningful search. It includes a new formal definition of a near duplicate, a pattern-based, and the proof of its completeness. Moreover, we demonstrate the results of experimenting on a collection of documents of several industrial projects.
\end{abstract}
}]

{~}

\vspace{-48pt}
\blfootnote{This work is partially supported by RFBR grant \mbox{16-01-00304.}}

\section{Introduction}
Software documentation quality issues have been studied since the 70's~\cite{Brooks1965en}, and continue to be addressed nowadays~\cite{Parnas2011}. In addition to that, both documentation and software are becoming more and more complex and require more resources for their development and maintenance.

 Copy/paste is commonly used in creating and modifying documents: a fragment of text is copied multiple times, and then edited and expanded to suit the subject.  The  use of this technique is justified by the fact that many software features described in documentation reuse functionality~--- this is true for requirements, user interface elements, tests, source code, etc. However, without the aid of specialized tools, repeatedly copied fragments create additional difficulties during maintenance because they require extensive synchronisation of changes in corresponding software features. Software reuse is a considerably more advanced research area than software documentation reuse. There is a  multitude of studies on software reuse and a part of them has been adopted by the industry (see surveys~\cite{Bassett1996,Jarzabek2006,ReuseReview2018}). However, the problem of software documentation reuse largely remains a subject of academic research only~\cite{Horie2010,Oumaziz2017,Koznov2008,Romanovsky2008,Jarzabek2017}.
We should also mention the problem of documentation unification~--- if there is a large volume of similar information, it is only reasonable to have it presented in a consistent manner. Therefore, duplicate detection and duplicate examination are important for setting up documentation reuse and unification.

 Duplicates in software documentation have been extensively studied during the last decade~\cite{Horie2010,Juergens2010,Nosal2014,Nosal2016,Wingkvist2010,Koznov2015,Luciv2018en.2,Koznov2017}. At the same time, there are no specialized tools for duplicate detection. Generally, various text search tools are used for this purpose. However, these tools cannot be employed in detection of near duplicates, i.e. text fragments with a substantial common part and certain variations. For this reason, we have created Duplicate Finder~\cite{DRTru,Luciv2018.1}. In~\cite{Luciv2018en.2}, we have presented a near duplicate search algorithm which considers near duplicates as a combination of exact duplicates. However, the output quality of this algorithm was low because it does not assess the meaningfulness of found duplicates.

In this paper we present an an approach for interactive detection of near duplicates. We involve the user in order to provide meaningfulness of the search process. In short, this process is organized as follows. At first, we automatically create a map of exact duplicates of the document using Clone Miner~\cite{basit2009}. 

The next step relies on the following assumption: an accumulation of exact duplicates in a certain place of the document points to a possible emergence of a near duplicate. At this step, the user moves to the most duplicate-populated document section using this map as a clue. After that, they select the text fragment that contains the most frequently used exact duplicates. Then, the user transforms this fragment into a full description of a certain software feature by extending its bounds. This way, the user ensures the meaningfulness of the fragment. Next, this description (textual string) is used as a pattern for further search. The user can also edit the search results by filtering out false positives and ensuring the meaningfulness of found fragments by expanding or narrowing their bounds in the document. 

 This article is organized as follows. In Section 2 we present an overview of existing studies that are similar to the current work, and in Section 3 we describe the approaches, methods, and technologies we have used in our study. Section 4 contains the description of the interactive near duplicate search process. In Section 5 we propose a new formal definition of a near duplicate, and in Section 6 we present the pattern-based near duplicate search algorithm that forms the basis of the proposed process. In addition, we formulate the criterion of completeness for the algorithm and prove that the proposed algorithm is complete. In Section 7 we provide complexity estimates for the algorithm, and in Section 8 we demonstrate the results of an experimental evaluation. 

\section{Related Work}
Duplicates in software documentation have been extensively studied during the last decade. Horie et al.~\cite{Horie2010} consider duplicates in API documentation of Java projects, extending JavaDoc with reuse tools. This approach is expanded with consideration of near duplicates by Nos{\'a}l' and Porub{\"a}n in~\cite{Nosal2014}. Similarly to approaches presented in references~\cite{Romanovsky2008, Koznov2008, Jarzabek2017}, the authors of this paper employ parametrization for defining variative parts of duplicates. However, this study does not consider the task of near duplicate search itself, and their definition of a near duplicate is informal. 

In their further research~\cite{Nosal2016}, they examine exact duplicates in embedded documentation of several open-source projects, but do not consider near duplicates. 

Wingkvist et al.~\cite{Wingkvist2010} use duplicates to evaluate the quality of documentation, with no consideration of near duplicates. 

The paper~\cite{Juergens2010} by Juergens et al. is an examination of duplicates in requirement specifications: the authors have analyzed 28 industrial documents, manually filtering and classifying found duplicates. The meaning of these duplicates was discussed (with emphasis on duplicates that corresponded to code clones). They also have  studied the influence of redundancy on the speed of reading and understanding texts. This work does not consider other types of software documentation, as well as near duplicates, although the authors do mention their existence. 

Ouzmazis et al.~\cite{Oumaziz2017} analyze API documentation of several well-known open source projects, classify detected duplicates, and consider the problem of documentation reuse. They do not consider near duplicates, however, they note that those duplicates occur quite often and are important in practice. 

Rago et al.~\cite{Rago2016} present near duplicate search in textual use case descriptions with the use of natural language processing methods. However, they consider a highly specific type of requirement specifications, which is rarely used in practice. Moreover, it is unclear how to apply this method to other types of documentation. 

Concluding our overview, we should note that most existing approaches except~\cite{Rago2016} use token-based tools of code clone analysis. This fact seriously complicates near duplicate detection. However, some authors acknowledge the existence and importance of near duplicates in documentation redundancy analysis and documentation reuse~\cite{Oumaziz2017,Juergens2010,Nosal2014}.

The approach presented in this work is largely based on pattern matching. This problem is well-studied and has been solved in multitude of ways for different contexts. Let us provide a short overview of this problem in the context of text search. 

The algorithm proposed by Ukkonen~\cite{Ukkonen1985} makes efficient matching approximate occurrences of pattern in text possible, but it requires a costly preprocessing of the pattern. 

Broder~\cite{Broder1997} describes a method for matching approximate occurrences of pattern in text using information retrieval methods; we should note that this approach also requires expensive preprocessing of input data (both document and pattern). 

Algorithms presented in~\cite{Ukkonen1985,Wu1992,Landau1988,Myers1999} are efficient but quite sophisticated, which complicates their use and modification, as well as proving their formal properties.  

Ukkonen's algorithm~\cite{Ukkonen1985} is suited for working with an immutable pattern, and Broder's approach~\cite{Broder1997} operates on an immutable document. Both of these situations are irrelevant to our task. Studies by Landau and Vishkin~\cite{Landau1988}, Myers~\cite{Myers1999} describe algorithms for detecting text fragments for which Levenshtein distance~\cite{Levenshtein1965en} does not exceed a pre-defined threshold. We should emphasize the high computational complexity of Levenshtein distance calculation, which, as shown by our experiments, makes this approach unsuitable for duplicate detection. A more detailed review of approximate pattern matching in text can be found in~\cite{Smyth2006en}. Using ideas proposed in this research area, we have created our own pattern matching algorithm while adhering to the following requirements: 
\begin{itemize}
	\item[(i)] the algorithm needs to perform near duplicate detection in accordance with our definition of near duplicates;
	\item[(ii)] we wanted to formally prove a number of properties of this algorithm;
	\item[(iii)] the run time has to be adequate because the algorithm is run in an interactive mode;
	\item[(iv)] the algorithm needs to yield as few false positives as possible.
\end{itemize}   

\section{Background} 
\subsection{Edit distance}
We use \emph{edit distance}~\cite{Levenshtein1965en} to determine the degree of similarity of two text fragments (strings of text). This distance is essentially the number of string editing operations required to convert one string into another: the less operations, the more similar the strings. Different definitions of edit distance differ by their admissible operations. In our work, we use \emph{longest common subsequence distance}~\cite{Bergroth2000,Leskovec2014} that uses only insertion and deletion of a symbol due to its suitability for near duplicate model described below, and, consequently, the convenience of further proofs. The authors of~\cite{Gusfield1997} prove that this type of edit distance has metric properties. Further on, we will denote the longest common subsequence distance between two strings $s_1$ and $s_2$ as $d(s_1,s_2)$.

Computing edit distance is a resource-consuming task. The complexity of the algorithm we have selected is estimated~\cite{Ratcliff1988} as $\mathcal {O}(|s_1 |*|s_2 |)$ in the average case. Furthermore, the authors of reference~\cite{Abboud2015} show that it is impossible to design an algorithm that can provide better complexity estimates for the worst case. In this work, we use the difflib library~\cite{pyDiffLib}, which is included in the standard Python package (performance-critical parts of which are implemented in C).

\subsection{Detecting exact duplicates with Clone Miner}
 We have employed exact duplicate detection to create a duplicate map, using which the user could select a pattern for matching. We have selected Clone Miner~\cite{basit2009} for this task, which is a software code clone detection tool. Clone Miner is a token-based tool, and it does not employ an abstract syntax tree. A token is a stand-alone word (sequence of symbols) in a document, separated from adjacent words by such delimiters as ``.'', ``,'', `` '', ``('', ``)'', etc. For example, the fragment ``FM registers'' consists of two tokens. Clone Miner considers text as an ordered collection of tokens and detects duplicate fragments (clones) using algorithms based on suffix trees~\cite{Abouelhoda2004}. We have chosen Clone Miner because it is easily integrated with other tools through command line interface and supports the Russian language.

\section{The Process}
The general purpose of our process is to ensure meaningfulness of duplicate detection via user interaction. A diagram which describes the workflow of the process is presented in Fig.~\ref{fig:methodics}.
\begin{figure}[h]
\includegraphics[width=85mm,page=2]{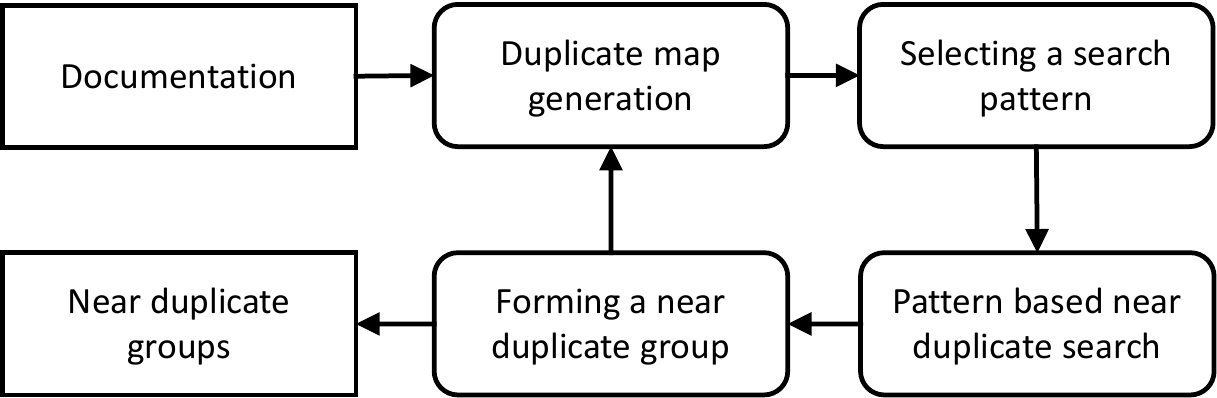}
\caption{Process overview}\label{fig:methodics}
\end{figure}
Let us describe the process in detail.

\paragraph{Generating a duplicate map.}
Using Clone Miner~\cite{basit2009}, all exact duplicate groups in the document are detected. Every token (word) $t$ is assigned a color from an RGB interval from white to red: $color(t)=(h(t)/T_m)*R+(1-(h(t))/ T_m)*W,$ where $R = [1, 0, 0]$, $W=[1, 1, 1]$, $h(t)$ is the exact duplicate group that has the maximum cardinality and contains $t$ (further on called token temperature), and $T_m$ is the maximum cardinality of an exact duplicate group in the document (maximum token temperature).\footnote{We only consider groups that consist of fragments longer than four tokens, because, according to our experiments~\cite{Koznov2015}, this particular constraint filters out many false positives.} The closer a token's color to red, the ``warmer'' it is. This metaphorical representation is called a heat map~\cite{Spakov2007}, an example of which can be seen in Fig.~\ref{fig:heatmap} 
\begin{figure}[h]
\centering\includegraphics[scale=0.60]{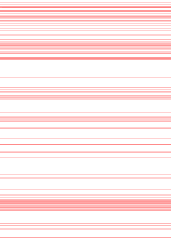}
\caption{Duplicate map}\label{fig:heatmap}
\end{figure}

The generated heat map provides an overview of the potential near duplicate occurences. The areas most likely to contain near duplicates are represented by red areas of different hue. Tokens that occur in the reddest areas are repeated the same or roughly the same number of times. Therefore, the probability that these tokens would form a meaningful near duplicate is quite high. This way, one may hope to obtain meaningful near duplicates that appear a significant number of times in the document.

\paragraph{Selecting a search pattern.} 
The user moves to the reddest (the ``warmest'') area on the heat map (Fig.~\ref{fig:heatmap}), zooms in on it, and selects a fragment (pattern) for further search (see Fig.~\ref{fig:reusemap}). During this process, the user does not only consider the color of the selected fragment, but aims to select a fragment that describes a software feature in full. To achieve this, the user can either include a white-colored text fragment in the pattern, or not include a red-colored one. Consider an example. Following the information from Fig.~\ref{fig:reusemap}, we select this fragment:

{\small \begin{Verbatim}[commandchars=\\\{\}]
To alter the owner, you must also be a \\
direct or indirect member of the new \\
owning role, and that role must have \\
CREATE privilege on the table's schema.\\
(These restrictions enforce that alte- \\
ring the owner doesn't do anything you \\
couldn't do by dropping and recreating \\
the table. However, a superuser can \\
alter ownership of any table anyway.)
\end{Verbatim}
}

 This fragment describes an integral software feature concerning the administration of the PostgreSQL DMBS: to alter the owner of a database table, you must also have specific rights or be an administrator.

\begin{figure}[h]
\centering\includegraphics[width=8.5cm]{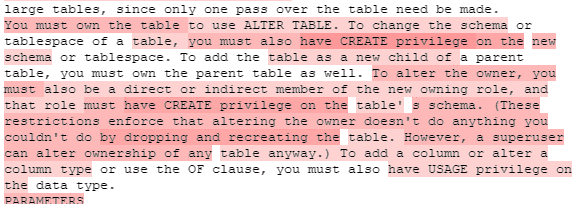}
\caption{Selecting a search pattern (PostgreSQL documentation)}\label{fig:reusemap}
\end{figure}

\paragraph{Near duplicate search.} The user selects a similarity measure for the highlighted fragment, which is a number from $1 / \sqrt{3}$ to~1, and launches the pattern matching algorithm\footnote{The $1 / \sqrt{3}\approx 0.577$ value was selected for the convenience of the following proofs; following our experiments, we have concluded that if the similarity measure is less than $1/2$, then, for smaller patterns (up to 15-20 tokens), the algorithm produces many non-meaningful matches; the lower bound we have selected is insignificantly larger than $1/2$.}.

\paragraph{Forming a near duplicate group.}
Having received the algorithm's output, the user modifies it. During the process the user deletes elements (near duplicate occurences) that only resemble the pattern syntax-wise but not meaning-wise. Furthermore, for each occurence the user can modify the bounds of the fragments to ensure the meaningfulness of each.

\section{Defining a near duplicate group}
In this section we generalize the definition of a near duplicate group that we have proposed earlier in~\cite{Luciv2018en.2,Luciv2017}. In contrast with the previous definition, here we use a parameter instead of a constant to define the similarity measure, and we allow to place extension points at the ends of duplicates. We will consider a document $D$ as a finite sequence of symbols, denoting its length as \(\mathrm{length}(D)\).

\begin{definition}
A \textbf{text fragment} is an occurrence of a certain symbol string in document $D$.\end{definition}
Therefore, for every text fragment $g$ of document $D$ there is an integer interval $[b, e]$, where $b$ is the position of the first symbol of the fragment, and $e$ is the position of the last. By $g \in D$, we denote a text fragment $g$ of document $D$. Next, let $[g]$ be the function that maps a text fragment $g$ to its interval, and let \(\mathrm{str}(g)\) be the function that maps a text fragment $g$ to its textual content. By $\mathrm{b}(g)$ and $\mathrm{e}(g)$ we denote the positions of the beginning and the end of $g$. Next, \(|g|\) is a function that takes a text fragment $g$ and returns its length as \(|g| = 1 + \mathrm{e}(g)-\mathrm{b}(g)\). Finally, we introduce a two-place predicate $\mathrm{Before}(g_1, g_2) $, which is true if and only if $\mathrm{e}(g_1) < \mathrm{b}(g_2)$.

\begin{definition}\label{def:ndgk}\textbf{Near duplicate group.}
Consider a collection of text fragments \(g_{1},\ldots,g_{M}\) of document \(D\). We will call this collection a near duplicate group with the similarity measure $k \in (1/\sqrt{3}, 1]$ (or simply a near duplicate group) if the following conditions are satisfied.

\begin{enumerate}
\item
$\forall i \in \{1, \ldots, M-1\} \;$ holds $\mathrm{Before} (g_{i}, g_{i + 1})$\;
\item
There exists a ordered collection of strings $(I_{1}, \ldots, I_{N})$ such as there is an occurrence of this collection in every text fragment, i.e. $\forall j \in \left\{ 1,\ldots,M \right\}\forall i \in \{ 1, \ldots, N \}\;
I_i\subset \mathrm{ str}(g_j)$ and $\forall i \in \{ 1, \ldots, N-1 \}\;$ holds $\mathrm{Before}(I_{i}^{j}, I_{i+1}^{j})$, where $I_{i}^{j}$ is an occurrence of $I_{i}$ in $g_{j}$, and the following condition is satisfied:
\begin{equation*}
\forall j \in \{ 1,\ldots, M \} \; \frac{\sum_{i = 1}^{N} |I_i|}{| g_j|} \geq k
\label{eq:ndg2}\end{equation*}
\end{enumerate}
\label{def:ndg2}
\end{definition}

We define the \emph{archetype} of a given group as a collection of strings \((I_{1},\ldots,I_{N})\). It is easy to show that the definition proposed above generalizes the definition given in ~\cite{Luciv2018en.2,Luciv2017}. If $G$ is a near duplicate group, then by $|G|$ denote the number of elements of this group.

\begin{definition}
Consider a text fragment \(p\) of document \(D\) (\(p \in \text{D}\)) and \(g \in D\). We say that \(g\) is a near duplicate of \(p\) with similarity \(k\), if \(g\) and \(p\) form a near duplicate group with similarity \(k\) defined according to Def.~\ref{def:ndg2}. \label{def:ndgp} \end{definition}

\section{Pattern based near duplicate search algorithm}
\renewcommand{\algorithmcfname}{\iflanguage{russian}{Алгоритм}{Algorithm}}

\SetInd{0.125em}{0.75em}
\begin{algorithm}[h]
\caption{\iflanguage{russian}{Поиск неточных повторов по образцу}{Pattern based near duplicate search algorithm}}
\LinesNumbered
\DontPrintSemicolon
\SetKwInput{KwIn}{\iflanguage{russian}{Входные данные}{Input data}}
\SetKwInput{KwOut}{\iflanguage{russian}{Результат}{Result}}
\KwIn{$D$ --- document,\\ $p$ --- pattern, $k$ --- similarity measure}
\KwOut{$R$}

\tcp{Phase 1 (scanning)}
 
$W_1 \leftarrow \emptyset$

\For{$\forall w_1: w_1 \in D \wedge |w_1| = L_w$}{
	\If{${\rm d}_{di}(w_1, p) \le k_{di}$}{
		add $w_1$ to $W_1$
	}
}

\tcp{Phase 2 (``shrinking'')}

$W_2 \leftarrow \emptyset$

\For{$w \in W_1$}{
  $w'_2 \leftarrow w$
  
  \For{$l \in I$}{
  	\For {$\forall w_2: w_2 \subseteq w \wedge |w_2| = l$}{
  		\If{${\rm Compare}(w_2, w'_2, p)$}{
  			$w'_2 \leftarrow w_2$
  		}
  		
  		add $w'_2$ to $W_2$
  	}
  }
}

\tcp{Phase 3 (filtering)}

$W_3 \leftarrow {\rm Unique}(W_2)$

\For{$w_3 \in W_3$}{
	\If{$\exists w'_3 \in W_3: w_3 \subset w'_3$}{
		remove $w_3$ from  $W_3$
	}
}

$R \leftarrow W_3$
\end{algorithm}\label{alg:main}

\subsection{Algorithm description}
 The algorithm consists of three phases. At \textbf{phase 1 (scanning)}, document \(D\) is scanned by a sliding window \(w\) of size \(L_{w} = \frac{\left| p \right|}{k}\) with a one symbol step\footnote{Here and further we do not round the lengths of the intervals to integers to save up space. Nevertheless, all proofs can be performed with rounded values as well.}. The text fragment that corresponds to the current window position is compared to pattern \(p\) using edit distance, and if they are close, i.e. \(d\left( p,w \right)\ \leq k_{\mathrm{di}}\), then this fragment is saved in the set \(W_{1}\). The threshold value $k_{\text{di}}$ is defined as follows:

\begin{equation}
k_{\text{di}} = \ \left| p \right|\left( \frac{1}{k} + 1 \right)\left( 1 - k^{2} \right)
\label{eq:kdi}\end{equation}

This choice will be explained below.

At \textbf{phase 2 (``shrinking'')}, we search for the largest text fragment that is closest to pattern \(p\) in every element of \(W_{1}\). Essentially, during this phase lengths of elements of \(W_{1}\) decrease, i.e. text fragments are ``shrunk''. This is reasonable since the window (and consequently, all elements of \(W_{1}\)) is of maximum possible size of a near duplicate of \(p\) (see lem.~\ref{lm:k1k}). During ``shrinking'' for every \(w_2 \in W_1\), all of its internal fragments are iterated over, starting with fragments of \(|p|*k\) length up to fragments of \(\frac{|p|}{k}\) length. The one that is closest to the pattern in terms of edit distance is selected. If there are several such fragments, the longest one should be taken. This phase results in the set \(W_{2}\).

At the \textbf{phase 3 (filtering)}, duplicate elements in \(W_{2}\) are eliminated.  They emerge at the previous phase because \(W_{1}\) can contain text fragments that differ by a window shift of several symbols. Furthermore, elements that are fully contained in other elements of \(W_{2}\) are filtered out. This phase results in the set \(W_{3}\) which is the output of the algorithm, i.e. the set \(R\). 

Let us describe the auxiliary functions used in algorithm 1. The \(\text{Compare}\) function is used during phase 2 to identify the text fragment which is closer to the pattern \(p\) in terms of edit distance. If the distance from both fragments to the pattern is the same, the longest fragment is selected:

\begin{equation*}
\begin{gathered}
\mathrm{Compare}(w_1,w_2,p)
= \\
= \begin{cases}
true & d(w_1,p) < d(w_2,p)\\
false & d(w_1,p) > d(w_2,p)\\
|w_1| > |w_2| & d(w_1,p) = d(w_2,p)
\end{cases}
\end{gathered}
\end{equation*}
The \(\text{Unique}\) function receives a collection of text fragments, iterates over it and discards duplicate fragments.

\subsection{Algorithm Completeness}
\textbf{The criterion of completeness} for our pattern based near duplicate search algorithm is defined as follows. The algorithm is complete if for arbitrary \(D\), $p \in D$, output of the algorithm \(R\), and for any near duplicate group \(G\) of fragment \(p\) with similarity \(k\) (see def.~\ref{def:ndgk}), the following condition holds true: 

\begin{equation}
\forall g \in G\ \exists w \in R: \;
\left| g \cap w \right | \geq O_{\min}(k)
\label{eq:criteria},\end{equation}
where \(O_{\min}(k)=\frac{\left| p \right|}{2}\left( 3k - \frac{1}{k} \right)\). This criterion can be explained as follows: for any fragment of document D that is a near duplicate of pattern \(p\), the  set \(R\) will contain a text fragment that significantly intersects with this near duplicate, allowing the user to easily recognise this duplicate in the output. The ratio of the intersecting portion to the whole pattern is bounded from below by the \(O_{\min}\ (k)\) function. \(O_{\min}(\frac{1}{\sqrt{3}}) = 0,\)  and for larger values of \(k\) \(O_{\min}\left( k \right) > \ 0\). This is true since the function increases with increasing \(k\)~--- its derivative is \(\frac{\left| p \right|}{2}\left( 3 + \frac{1}{k^{2}} \right)\)  and it is obviously positive for all \(\frac{1}{\sqrt{3}} < k \leq 1\). In practice, the best results are achieved for \(k \geq 0.77\): for these values \(O_{\min}\left( k \right) > \frac{\left| p \right|}{2}\), i.e. all elements of \(R\) intersect all near duplicates at least by half of the pattern's length. Note that the lower estimate $O_{\min}(k)$ is pessimistic: the experimental results demonstrate a larger overlap of the output and the near duplicates contained in the document. Let us continue on to the completeness of the proposed algorithm, proving several auxiliary propositions first.

\begin{lemma}
Let \(G\) be a near duplicate group of fragment \(p\) with similarity \(k\). Then \(\forall\ g_{1},\ g_{2} \in G\)
\(k \leq \frac{\left| g_{1} \right|}{\left| g_{2} \right|} \leq \frac{1}{k}\) holds true.
\label{lm:k1k}\end{lemma}
\begin{proof}
Suppose \(A\) is the archetype of group \(G\). Then 
\(k|g_{1}| \leq |A|\) and \(k | g_2 | \leq |A|\).
Because
\(A \subset \mathrm{str}(g_1)\ \text{and}\ A \subset \mathrm{str}(g_2)\),
we have:
\(k|g_1| \leq |A| \leq |g_1| \ \text{and}\ k|g_2| \leq |A| \leq |g_2|\).
Therefore, \(k|g_1| \leq |g_2| \ \text{and}\ k|g_2| \leq |g_1|\).
Dividing these inequalities by \(|g_2|\) and \(k|g_2|\) respectively, we get the required result.
\end{proof}

\begin{lemma}
Let \(G\) be a near duplicate group of fragment \(p\) with similarity \(k\).
Then \(\forall\ g\ \in G\) the following holds true:
\(d(g,\ p) \leq \left( 1 - k^{2} \right)\left| p \right|\).
\label{lm:p1k2}\end{lemma}
\begin{proof}
Because \(p\) and \(g\) belong to the same near duplicate group, they have the same archetype and can be presented in the following way:

\begin{equation*}
\begin{gathered}
p = v_{0}^{p}I_{1}v_{1}^{p}I_{2}\ldots v_{N - 1}^{p}I_{N}v_{N}^{p},\\
g = v_{0}^{g}I_{1}v_{1}^{g}I_{2}\ldots v_{N - 1}^{g}I_{N}v_{N}^{g},
\end{gathered}
\end{equation*}
where \(I_{1},I_{2}\ldots{,I}_{N}\) is the archetype of group \(G\),
\(v_{0}^{p},v_{1}^{p},\ldots,v_{N}^{p}\)  is the variative part of \(p\), and 
\(v_{0}^{g},v_{1}^{g},\ldots,v_{N}^{g}\) is the variative part of
\(g\).

Let us introduce the following notations:
\(v^{p} = v_{0}^{p}v_{1}^{p},\ldots,v_{N}^{p},\ v^{g} = v_{0}^{g}v_{1}^{g},\ldots,v_{N}^{g},\ A = I_{1}I_{2}\ldots I_{N}\).
Then according to~(\ref{eq:ndg2}), we have
\(|A|/|p| \ge k \Rightarrow |p| - |v^p| \geq |p|k \Rightarrow |p| - |p|k \geq |v^p| \Rightarrow
|p| (1 - k) \geq v^p\), and, likewise, \(|g|(1 - k) \geq v^g\). Moreover, \(g\) can be obtained from \(p\) by substituting \(v_i^p\) for \(v_i^g\), i.e.
\(d(g, p) \leq |v^g| + |v^p| \leq (1 - k)(|p| + |g|)\).
According to lemma~\ref{lm:k1k} we have \(|g| \leq k |p|\).
Then \(d(g,p) \leq (1 - k)(1 + k)|p| = (1 - k^{2})|p|\).
\end{proof}

\begin{lemma}
For any \(p \in D\), \(k \in \left(1/\sqrt{3}, 1\right]\), near duplicate group \(G\) of fragment \(p\) with similarity \(k\) (def.~\ref{def:ndgp}), the criterion of completeness~\ref{eq:criteria} is satisfied in respect to the results of phase 1.
\label{lm:ph1}\end{lemma}
\begin{proof}
As mentioned above, the triangle inequality is satisfied for longest common subsequence distance:
\(d(\mathit{fr}, p) \leq d(\mathit{fr}, g) + d(g, p)\).
According to lemma~\ref{lm:p1k2},
\(d(g, p) \leq |p|( 1 - k^2)\).
We also know that \(g \subseteq fr\). Therefore, because we can obtain \(g\)
from \(\mathit{fr}\) by removing all symbols that belong to \(fr\setminus g\), 
$d(\mathit{fr}, g) \leq |\mathit{fr}| - |g|$ holds true.
But because \(|\mathit{fr}| = \frac{|p|}{k}\)
and according to lemma~\ref{lm:k1k}, \(|g| \geq k |p|\),
the following also holds true:
\(|\mathit{fr}| - |g| \leq \frac{1}{k}|p| - k|p|\).
Therefore,
\(d(\mathit{fr}, p) \leq |p|\left( \frac{1}{k} - k + 1 - k^{2} \right) = |p|\left( 1 + \frac{1}{k} \right)\left( 1 - k^2\right)\).
It is obvious that during the scanning on phase 1 there will be a state in which the window contains\(\mathit{fr}\). Then, according to (\ref{eq:kdi}),
\(fr \in W_{1}\). Therefore, the following holds true:
\[\forall g \in G: \left( |\mathit{fr}| = \frac{|p|}{k}, g \subseteq \mathit{fr} \right) \Rightarrow \mathit{fr} \in W_1.\]
Since for any near duplicate there is an element of \(W_{1}\) that does not only intersect with this duplicate, but contains it completely, criterion~\ref{eq:criteria} is satisfied for \(W_{1}\) if we consider it as the set \(R\).
\end{proof}

\begin{lemma}
 For any \(p \in D\), \(k \in \left(1/\sqrt{3}, 1\right]\) and near duplicate group \(G\) of fragment \(p\) with similarity \(k\)~(def.~\ref{def:ndgp}) the criterion of completeness~\ref{eq:criteria} is satisfied in respect to the output of phase 2.\label{lm:ph2}\end{lemma}
\begin{proof}
 We omit a formal proof due to its large size. The main idea here is considering the worst case where during ``shrinking'', the length of \(w_1 \in W_1\) elements is decreased to \(k|p|\). Considering corner cases (an element positioned in the center or right at the ends of \(w_1\)) allows us to confirm that the lemma holds true.
\end{proof}

\begin{lemma}
For any \(p \in D\), \(k \in \left(1/\sqrt{3}, 1\right]\) and near duplicate group \(G\) of fragment \(p\) with similarity \(k\) (def.~\ref{def:ndgp}) the criterion of completeness~\ref{eq:criteria} is satisfied in respect to the output of phase 3.\label{lm:ph3}\end{lemma}
\begin{proof}
Phase 3 consists of element deletion from \(W_2\) only. The intervals of the deleted elements are contained in intervals of other elements. It is obvious that if \(W_2\) satisfied the criterion, then \(W_3\) will satisfy it as well.
\end{proof}

\begin{theorem}
The criterion of completeness is satisfied for any \(p \in D\), \(k \in \left(1/\sqrt{3}, 1\right]\), corresponding algorithm output \(R\) and any near duplicate group \(G\) of fragment \(p\) with similarity \(k\).
\label{th:recall}\end{theorem}
\begin{proof}
The output of phases 1--3 was proven to satisfy the criterion~\ref{eq:criteria}  in lemmas~\ref{lm:ph1},~\ref{lm:ph2},~\ref{lm:ph3}.
\end{proof}

\subsection{Optimizing the algorithm}
The proposed algorithm turned out to be inadequate performance-wise: its run time exceeded one hour when searching for patterns larger than 100 symbols in documents of about 2 MB in size. Furthermore, the algorithm produced many false positives~--- its output contained the same text fragments that were insignificantly shifted relatively to each other. As the result, a range of optimizations has been suggested.

\textbf{Optimization 1} is applied during phase 1 (scanning). It allows to reduce the number of calculations of $d$, significantly improving the run time of the algorithm. It is based on the known Boyer-Moore algorithm, which is intended for matching a pattern in a string~\cite{Boyer1977}: during the scan, a check is performed to see by how many symbols the window can be shifted without skipping the required result. Therefore, at each step of the scan we check whether \(d\left( w,p \right) > k_{\text{di}} + 1\) (\(w\) is the window position) holds. If it is true, then we slide the window by \((d\left( w,p \right) - k_{\text{di}})/2\) symbols to the right. Otherwise, we slide it by one symbol.

\textbf{Optimization 2} is applied during phase 2 (``shrinking''). It allows to reduce the number of computations of $d$ as well. The approach is similar to the one used in the previous optimization. During ``shrinking'' of a text fragment \(w_{1}\), the window scans the fragment in a symbol-by-symbol manner. At each step \(d\left( p,\ w_{2}^{'} \right)\) is computed and, if necessary, its minimum value \(d_{\min}\) is updated. If for the current window position \(w_{2}^{'}\), \(d\left( p,\ w_{2}^{'} \right) > d_{\min} + 1\) holds true, slide the window to the right by \((d\left( p,\ w_{2}^{'} \right) - d_{\min})/2\) symbols. Otherwise, slide it by one symbol. The \(d_{\min}\) value is updated at the beginning of each iteration corresponding to the next value of the sliding window width.

\textbf{Optimization 3} is applied during phase 3 (filtering). It allows to minimize the cardinality of \(W_{3}\). It is as follows: the set is divided into maximum subsets that are transitively closed under intersection. Further, for every such subset a \(w_{3}\) fragment with the minimum value of \(d\left( w_{3},p \right)\) is selected, or if there are several such fragments, the one with maximum length. All remaining elements of the set are deleted.

\textbf{Optimization 4} is applied during phase 3, extending all text fragments of \(W_{3}\) up to complete words. The bounds of a text fragment can ignore the bounds of words, i.e. incomplete words can be included into text fragments. In order to address this, text fragments are expanded to include these words fully. This helps to decrease the number of false positives in the algorithm's output.

\textbf{Optimization 5} is applied during phases 1 and 2. Its purpose is reusing $d$ for the same strings and parallelizing the ``shrinking'' of the elements of \(W_1\).

Let us show how these optimizations affect the algorithm's completeness.

\begin{theorem}
Optimizations 1, 2, 4, 5 preserve the completeness property.\label{th:optrecall} \end{theorem}
\begin{proof}
Consider two strings that are results of concatenation: \(s_1=ab\) and \(s_2=bc\), where \(|a|=|c|\) and \(d(s_1, s_2) = d\). We can easily show that \(|a|+|c| \ge d\). Using this fact, it is easy to prove the completeness of optimization 1. The completeness of optimization 2 is proven in the same way. The completeness of optimization 4 can not be doubted because it only extends the elements of the output. Finally, optimization 5 is complete because it only considers the implementation of the algorithm.
\end{proof}

\begin{remark}
Situations where optimization 3 does not satisfy the criterion of completeness are possible, but our experiments show that their number is insignificant in practice.
\end{remark}

\section{Algorithm complexity}
Document length $|D|$, pattern length $|p|$, the $k$ value, and the cardinality of the near duplicate group of the pattern $|G_p|$ are all significant parameters that influence the run time of the algorithm. Let us estimate the algorithm's complexity depending on these parameters.

The average complexity of calculating $d$ (i.e. edit distance) is \(\mathcal{O}(|p|^2)\)~\cite{Ratcliff1988}. Consequently, the average complexity of phase 1 is proportional to $|p|^2$~and $|D|$. During phase 2 all of the internal fragments of each $w_1 \in W_1$ are iterated over, and it is easy to show that their number is proportional to $|p|$. Furthermore, the cardinality of the set $W_1$ is proportional to $|G_p|$ and $k_{di}$. Finally, the complexity of phase 2 is $\mathcal{O}(|G_p|)$ and $\mathcal{O}(|p|^4)$. The complexity of phase 3 operations is $\mathcal{O}(|W_2|*\log|W_2|)$, but because $|W_2| = |W_1|$, the complexity of phase 3 is $\mathcal{O}(|G_p|*\log |G_p|)$ and $\mathcal{O}(|p|*\log|p|)$. Optimizations 1 and 2 on average lead to ``skips'' during iteration, the size of which is proportional to $k_{di}$ (and hence $|p|$), making the complexity of phases 1 and 2 $\mathcal{O}(|p|)$ and $\mathcal{O}(|p|^3)$ respectively.
Therefore, with $k = \mathit{const}$ the algorithm's run time can be estimated as $\mathcal{O}(|D|)$, $\mathcal{O}(|p|^3$)~and $\mathcal{O}(|G_p|*\log |G_p|)$.

\begin{theorem}
The complexity of the algorithm with fixed $D$ and $p$ is estimated as $\mathcal{O}(1/k^4)$ on average.
\label{th:complexity_k}\end{theorem}
We omit the proof due to its large volume. 

\section{Evaluation}
Theoretical complexity estimates are not sufficient for determining the real run time of the proposed algorithm. These estimates were produced using certain significant parameters of the algorithm independently to simplify the proofs, while real complexity can depend on their combinations. Another argument for the necessity of experimental evaluation is the fact that theoretical estimates do not provide the real value intervals of these parameters. Finally, other properties of the algorithm need to be evaluated as well.

 We have conducted our experiments to answer the following questions: 
 \begin{enumerate}
 	\item[(i)] what is the run time of the pattern matching algorithm on real data;
 	\item[(ii)] how large are algorithm's outputs having real data as input.
 \end{enumerate}
   The first question is important because the algorithm is used in interactive mode, and therefore its run time should not exceed several minutes. Considering output volume, we have proven that our algorithm's output contains all existing near duplicates of a certain pattern. It is, however, unclear, how exactly large are the real outputs of the algorithm ~--- outputs that contain over 100 elements become more or less unfeasible for human analysis. In turn, output volume is affected by the number of false positive matches and the number of near duplicates in the document.

 We have experimented on 19 industrial documents both in Russian and English (described in reference~\cite{Luciv2018en.2}). The experiments were conducted on a computer with the following specifications:  Intel Core i7 2600, 3.4~GHz, 16~GB RAM. The documents were converted into ``flat text'' (UTF-8 encoding) with Pandoc~\cite{Pandoc}. After the convertation, the size of the documents ranged from 0.04 MB to 2.5 MB (0.75 MB on average). We are inclined to think that these numbers are realistic for $|D|$. However, we should note that it is necessary to create a more representative selection of different documentation types in order to obtain more precise estimates.

The experiments were conducted as follows. We have run the algorithm for patterns of length ranging from 50 to 1000 symbols with a 50-symbol step. A 1000-symbol fragment is about 25\% of a page of a docx document, i.e. it is a large fragment, and following our experiments, duplicates are significantly smaller in general. We have iterated the similarity measure value k from 0.6 to 1 with 0.1 step for each selected pattern in each document. We have selected the pattern in the following way. Having a fixed pattern length, we followed our process and selected the ``warmest'' area in the document of this length, calculating it automatically as a fragment where the following expression reaches its maximum value:
$\sum_{t \in fr} h(t)$. In this expression $t$ is a token of fragment $fr$ and $h(t)$ is its temperature. The sum is calculated over all tokens of the fragment, including possibly incomplete leftmost and rightmost tokens.

Analyzing the data obtained from the experiments to answer the question whether the algorithm's run time is suitable for interactivity, we have established the following: in 38\% of cases the algorithm ran for less than 5 seconds, in 78\% cases~--- less than 30 seconds, in 90\% of cases~--- less than 2 minutes. These run times are fairly adequate for interactive mode.

We have obtained the following data on the output volumes of our algorithm: 84\% of outputs contained less than 100 elements, 5\% outputs --- from 100 to 200 elements, 5.6\% --- from 200 to 600 elements, 5.4\% --- from 600 to 1000 elements. Thus, the majority of near duplicate groups in software documentation are relatively small (containing up to 100 elements), which follows from Theorem 1 and our experimental results.

\section{Conclusion}
In this study we have presented an interactive near duplicate search process for software documentation. This process solves the problem of meaningful extraction of near duplicates by involving the user, who can use an automatically generated heat map of exact duplicates to detect the most probable occurrences of near duplicates. We have created a pattern-based near duplicate search algorithm and provided optimizations for it. We have proven the completeness of the algorithm, meaning that all near duplicates contained in the document are present in the algorithm's output. More precisely, duplicates located in the document significantly intersect with particular elements of the output, and this is why the user can identify them with ease. Our process allows user to manually edit their bounds and to include them in the output in full. We present complexity estimates for our algorithm as well as experimental results. These results suggest that duplicate groups in software documentation generally do not exceed 100 elements, and the algorithm itself performs adequately for practical use.

In the future, we plan to study different types of software documentation in detail using our algorithm and experiment model (focusing on API documentation first). We also intend to thoroughly examine the behavior of our algorithm with varying input parameters (pattern length and similarity measure). Finally, we plan to switch to automatic methods of detecting meaningful duplicates via machine learning. A detailed analysis of different types of near duplicates in different software documentation types is required as well. Other fruitful areas for future work are integration of documentation reuse (in the context of requirement development) with automatic test development~\cite{Drobintsev2014,Pakulin2011}, and visualization of duplicate structure using diagrams~\cite{Gavrilova2011}.
\bibliographystyle{utf8gost71upatched}
\bibliography{references,webresources}
\label{lastpage}

\end{document}